\documentclass[11pt]{article}

\usepackage[T1]{fontenc}
\usepackage[utf8]{inputenc}
\usepackage{lmodern}
\usepackage[a4paper,margin=1in]{geometry}
\usepackage{microtype}
\usepackage{amsmath,amssymb,amsthm,mathtools}
\usepackage{booktabs}
\usepackage{enumitem}
\usepackage{xcolor}
\usepackage{tikz}
\usetikzlibrary{arrows.meta,positioning,fit,backgrounds}
\usepackage{hyperref}
\usepackage[nameinlink,noabbrev]{cleveref}
\crefname{theorem}{Theorem}{Theorems}
\Crefname{theorem}{Theorem}{Theorems}
\crefname{lemma}{Lemma}{Lemmas}
\Crefname{lemma}{Lemma}{Lemmas}
\crefname{definition}{Definition}{Definitions}
\Crefname{definition}{Definition}{Definitions}

\hypersetup{
  colorlinks=true,
  linkcolor=blue!55!black,
  citecolor=blue!55!black,
  urlcolor=blue!55!black,
  pdftitle={Tight Hardness for Temporal Path Covers at Vertex-Cover Number Two},
  pdfauthor={Alp Par}
}

\newtheorem{theorem}{Theorem}[section]
\newtheorem{lemma}{Lemma}[section]
\newtheorem{proposition}{Proposition}[section]
\newtheorem{corollary}{Corollary}[section]
\theoremstyle{definition}
\newtheorem{definition}{Definition}[section]

\newcommand{\TPC}{\textsc{TPC}}
\newcommand{\TDPC}{\textsc{TDPC}}
\newcommand{\DNMTS}{\textsc{DNMTS}}
\newcommand{\vc}{\operatorname{vc}}

\title{\textbf{Tight Hardness for Temporal Path Covers\\at Vertex-Cover Number Two}}
\author{Alp Par\\
Istanbul Topkapı University\\
\href{mailto:alppar@topkapi.edu.tr}{alppar@topkapi.edu.tr}\\
\href{https://alppar.com}{alppar.com}}
\date{}

\begin{document}
\maketitle

\begin{abstract}
Temporal Path Cover (\TPC{}) and Temporally Disjoint Path Cover (\TDPC{}) ask for minimum-cardinality covers of a temporal digraph by temporal paths, with \TDPC{} additionally requiring pairwise temporal disjointness. Cioni et al. proved both problems NP-hard on temporal DAGs whose underlying undirected graph has vertex-cover number three and left the vertex-cover-two case open. We resolve this question affirmatively for both problems. Using polynomial-time reductions from Distinct Numerical Matching with Target Sums, we show that \TPC{} and \TDPC{} are NP-complete already when the underlying undirected graph has vertex-cover number exactly two. For \TPC{}, this yields an exact vertex-cover threshold: polynomial-time solvability at vertex-cover number at most one and NP-completeness at two. For \TDPC{}, we further study the boundary case of temporal oriented stars. A reuse-normalization argument reduces optimum solutions to leaf-unique two-edge merges. If either the incoming or outgoing side is single-label, the optimum remains polynomial-time computable even when the opposite side is multi-label. We also introduce a span parameter for feasible merge realizations and prove that any feasible family of $k$ merges forces a realization of span at least $2k-2$. Hence \TDPC{} on temporal oriented stars is in XP parameterized by maximum span, and is polynomial-time solvable for every fixed span bound. The unrestricted multi-label star case remains open.
\end{abstract}

\section{Introduction and contributions}

Temporal Path Cover and Temporally Disjoint Path Cover were introduced by Chakraborty et al.~\cite{chakraborty2024} as temporal analogues of classical path-cover problems on directed acyclic graphs. Their results already reveal a sharp algorithmic contrast: \TPC{} is polynomial-time solvable on temporal oriented trees, whereas \TDPC{} remains NP-hard on that class. Cioni et al.~\cite{cioni2026} subsequently developed parameterized results for both problems and proved NP-hardness on temporal DAGs whose underlying undirected graph has vertex-cover number three. As of this writing, that work is available as the arXiv preprint arXiv:2607.00118v4. Their concluding discussion explicitly asks whether this bound is tight, namely whether either problem remains NP-hard when the vertex-cover number is at most two. Their positive FPT result for \TPC{} parameterized jointly by the vertex-cover number and the number of time steps is fully consistent with our hardness result: in our reductions the maximum time label $t_{\max}$ is not bounded by a constant parameter but grows polynomially with the source instance.

\subsection*{Related work and position of the result}
Temporal graphs have developed into a broad algorithmic framework in which classical notions of reachability, separation, and routing may change substantially once edge availability is time dependent; for an algorithmic overview, see Michail~\cite{michail2016}. Temporally disjoint routing was studied explicitly by Klobas et al.~\cite{klobas2021}, who showed strong hardness phenomena even for a small number of requested temporal paths, and by Kunz, Molter, and Zehavi~\cite{kunz2023}, who refined the picture under restrictions on the underlying graph. Related temporal separation problems likewise exhibit sharp structural complexity transitions~\cite{zschoche2020}. Against this background, Chakraborty et al.~\cite{chakraborty2024} introduced the path-cover problems considered here and established the first tractability and hardness frontier for temporal DAGs. Cioni et al.~\cite{cioni2026} then showed that both problems remain hard at constant vertex-cover number, with their construction attaining vertex-cover number three. The present work addresses exactly the remaining boundary value posed in their concluding open question.

We close that case for both problems. Our main results are the following.

\begin{theorem}[TPC at vertex-cover number two]\label{thm:tpc}
The decision version of \TPC{} is NP-complete on temporal DAGs $G$ satisfying
\[
\vc(U_G)=2.
\]
\end{theorem}

\begin{theorem}[TDPC at vertex-cover number two]\label{thm:tdpc}
The decision version of \TDPC{} is NP-complete on temporal DAGs $G$ satisfying
\[
\vc(U_G)=2.
\]
\end{theorem}

For \TPC{}, \cref{thm:tpc} is tight with respect to the vertex-cover number. Indeed, if $\vc(U_G)\le 1$, then after removing isolated vertices the underlying undirected graph is a star, hence an oriented tree; the polynomial-time algorithm of Chakraborty et al.~\cite{chakraborty2024} applies componentwise. We record the resulting dichotomy formally in \cref{cor:tpc-threshold}. For \TDPC{}, the exact status at vertex-cover number one is subtler. In \cref{sec:boundary} we develop a normalization framework for temporal oriented stars and prove two broader tractability results: the problem is polynomial whenever either side of the oriented star is single-label, even if the opposite side has arbitrarily many labels, and it is in XP when parameterized by maximum merge span. In particular, every fixed span bound yields a polynomial-time algorithm. The unrestricted two-sided multi-label star case remains open.

The two hardness reductions use different mechanisms. For \TPC, overlap between paths is allowed, so the construction enforces exact usage through a tight source-counting argument and filler paths. For \TDPC, temporal disjointness itself is used to lock paths to designated time blocks.

\begin{corollary}[Exact vertex-cover threshold for \TPC]
\label{cor:tpc-threshold}
On temporal DAGs, \TPC{} is polynomial-time solvable when $\vc(U_G)\le 1$ and NP-complete already when $\vc(U_G)=2$.
\end{corollary}

\begin{proof}
If $\vc(U_G)=0$, the graph is edgeless and the optimum is immediate. If $\vc(U_G)=1$, every edge of the underlying undirected graph is incident with one common vertex, so the non-isolated part is a star. Chakraborty et al.~\cite{chakraborty2024} give a polynomial-time algorithm for \TPC{} on temporal oriented trees; isolated vertices can be handled independently. NP-completeness at value two is \cref{thm:tpc}.
\end{proof}

\begin{table}[ht]
\centering
\caption{High-level comparison of the reductions.}
\begin{tabular}{@{}lll@{}}
\toprule
 & \TPC{} reduction & \TDPC{} reduction \\
\midrule
Source problem & Distinct NMTS & Distinct NMTS \\
Cover vertices & $x,y$ & $x,y$ \\
Main device & Cell matrix and fillers & Time blocks and guards \\
Global forcing & Majorization and total sum & Temporal conflicts and total sum \\
Underlying digraph & DAG & DAG \\
Target vertex cover & Exactly $2$ & Exactly $2$ \\
\bottomrule
\end{tabular}
\end{table}

\section{Preliminaries}

A temporal digraph is a pair $G=(U_G,\lambda)$, where $U_G=(V,E)$ is a directed graph and $\lambda(e)$ is the set of discrete time steps at which edge $e$ is available. All edges have unit traversal time. A strict temporal directed path is a vertex-simple directed path whose selected edge labels are strictly increasing. A temporal DAG is a temporal digraph whose underlying digraph $U_G$ is acyclic.

A temporal path cover is a collection of temporal paths whose union of vertex sets is $V$. In a temporally disjoint path cover, no two paths occupy the same vertex at the same time. Following Chakraborty et al.~\cite{chakraborty2024} and Cioni et al.~\cite{cioni2026}, if a path enters an internal vertex at time $r$ and leaves it at time $s$, then it occupies that vertex at every integer time in $[r+1,s]$. The first and last vertices are occupied at the departure and arrival time units specified in their definition.

For a static undirected graph $H$, let $\vc(H)$ denote the minimum size of a vertex cover. We apply this parameter to the underlying undirected graph obtained from $U_G$ by forgetting edge orientations.

\subsection{Source problem}

We reduce from Distinct Numerical Matching with Target Sums. Hulett, Will, and Woeginger proved that Numerical Matching with Target Sums remains strongly NP-complete when all $3n$ input elements are distinct~\cite{hulett2008}; see also the explicit formulation used by Cioni et al.~\cite{cioni2026}. Because the source problem is strongly NP-complete, we may assume that all numerical values in the instance, including the $a_i$, $b_j$, and $t_k$, are polynomially bounded in the input size. Consequently, every auxiliary quantity introduced below, including $M$, $C$, $D$, $H$, $L_j$, and $S_j$, and every temporal label derived from them, has polynomial encoding length. Each construction also contains only polynomially many vertices, arcs, and labels. Thus both mappings described below are polynomial-time many-one reductions producing instances of polynomial encoding size.

\begin{definition}[Distinct Numerical Matching with Target Sums]
An instance of \DNMTS{} consists of three sets of pairwise distinct positive integers
\[
A=\{a_1,\ldots,a_n\},\qquad
B=\{b_1,\ldots,b_n\},\qquad
T=\{t_1,\ldots,t_n\},
\]
satisfying
\[
\sum_{i=1}^{n}a_i+\sum_{j=1}^{n}b_j=\sum_{k=1}^{n}t_k.
\tag{1}\label{eq:balance}
\]
The question is whether there are permutations $\pi$ and $\sigma$ of $[n]$ such that
\[
a_{\pi(j)}+b_j=t_{\sigma(j)}
\qquad\text{for every }j\in[n].
\]
\end{definition}

Reordering elements within any of the three sets does not change the instance. We will use a strictly decreasing order for $A$ when needed and a strictly increasing order for $B$.

\section{NP-completeness of TPC}\label{sec:tpc}

We prove \cref{thm:tpc}. Sort the input sets so that
\[
a_1>a_2>\cdots>a_n,
\qquad
b_1<b_2<\cdots<b_n.
\tag{2}\label{eq:orders}
\]

\subsection{Construction}

Create two distinguished vertices $x$ and $y$. For each $j\in[n]$, create a main source $p_j$. For every $r\in[n]$ and $\ell\in[n-1]$, create a filler source $f_{r,\ell}$. For every $i,h\in[n]$, create a cell vertex $\alpha_{i,h}$, and for every $k\in[n]$, create a target vertex $\beta_k$.

There are $n+n(n-1)=n^2$ sources. Set
\[
K=n^2.
\tag{3}\label{eq:tpcK}
\]

Let
\[
M=\max\left\{\max_{i,h}(a_i+b_h),\max_k t_k\right\},
\qquad C=2n+10,
\qquad D=C+2M+10.
\]
Add the following arcs and label sets:
\begin{align}
\lambda(p_j,x)&=\{2j-1\} &&(j\in[n]), \tag{4}\label{eq:tpcmainin}\\
\lambda(f_{r,\ell},x)&=\{D+2r-1\} &&(r\in[n],\ \ell\in[n-1]), \tag{5}\label{eq:tpcfillin}\\
\lambda(x,\alpha_{i,h})&=\{2h,\ D+2i\} &&(i,h\in[n]), \tag{6}\label{eq:tpccell}\\
\lambda(\alpha_{i,h},y)&=\{C+2(a_i+b_h)\} &&(i,h\in[n]), \tag{7}\label{eq:tpccelly}\\
\lambda(y,\beta_k)&=\{C+2t_k+1\} &&(k\in[n]). \tag{8}\label{eq:tpcout}
\end{align}

\begin{figure}[ht]
\centering
\begin{tikzpicture}[>=Latex,node distance=12mm and 18mm,
  v/.style={circle,draw,minimum size=7mm,inner sep=1pt},
  box/.style={draw,rounded corners,minimum width=26mm,minimum height=13mm,align=center},
  lab/.style={font=\scriptsize,align=center}]
\node[v] (x) {$x$};
\node[box,left=24mm of x,yshift=9mm] (main) {main sources\\$p_j$};
\node[box,left=24mm of x,yshift=-9mm] (fill) {filler sources\\$f_{r,\ell}$};
\node[box,right=18mm of x] (cells) {$n\times n$ cells\\$\alpha_{i,h}$};
\node[v,right=18mm of cells] (y) {$y$};
\node[box,right=18mm of y] (targets) {targets\\$\beta_k$};
\draw[->] (main.east) -- node[above,lab] {$2j-1$} (x.west);
\draw[->] (fill.east) -- node[below,lab] {$D+2r-1$} (x.west);
\draw[->] (x) -- node[above,lab] {$2h$ or $D+2i$} (cells);
\draw[->] (cells) -- node[above,lab] {$C+2(a_i+b_h)$} (y);
\draw[->] (y) -- node[above,lab] {$C+2t_k+1$} (targets);
\end{tikzpicture}
\caption{Layered structure of the \TPC{} construction. Every edge is incident with $x$ or $y$.}
\label{fig:tpc}
\end{figure}
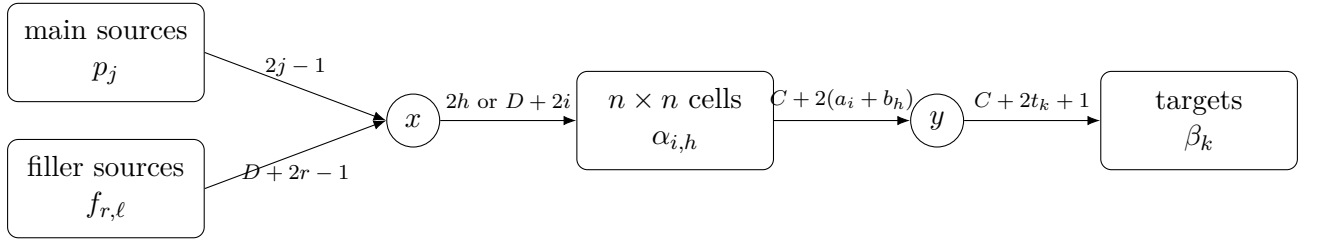

The underlying digraph has the topological order
\[
\{p_j,f_{r,\ell}\}\prec x\prec\{\alpha_{i,h}\}\prec y\prec\{\beta_k\}.
\]
Hence it is a DAG. Every edge is incident with $x$ or $y$, so $\vc(U_G)\le2$. The undirected edges $p_1x$ and $\alpha_{1,1}y$ are vertex-disjoint, so no single vertex covers all edges. Thus $\vc(U_G)=2$.

\subsection{Tight structure}

\begin{lemma}\label{lem:tpc-tight}
Every temporal path cover of size at most $K$ has exactly $K$ paths. Every main-source path contains one cell and one target, every filler-source path contains one cell, and every source, cell, and target vertex occurs on exactly one path.
\end{lemma}

\begin{proof}
Let $I=V(G)\setminus\{x,y\}$. The graph has exactly $K=n^2$ indegree-zero sources, and no directed path can contain two sources. Therefore every cover has at least $K$ paths.

The set $I$ contains $n^2$ sources, $n^2$ cells, and $n$ targets, hence
\[
|I|=2n^2+n.
\]
A path beginning at a main source contains at most three vertices of $I$:
\[
p_j\to x\to\alpha_{i,h}\to y\to\beta_k.
\]
A filler enters $x$ after time $D$, whereas every label on an edge $\alpha_{i,h}\to y$ is at most $C+2M<D$. Thus a filler path cannot continue beyond its cell and contains at most two vertices of $I$.

Consequently, the $n$ main-source paths and the $n(n-1)$ filler-source paths contain at most
\[
3n+2n(n-1)=2n^2+n=|I|
\]
occurrences of vertices in $I$.

For each source path $P$, let $q(P)=|V(P)\cap I|$. The preceding path-type bounds give
\[
\sum_P q(P)\le |I|.
\]
On the other hand, the paths cover every vertex of $I$, and therefore subadditivity of cardinality gives
\[
|I|
=
\left|\bigcup_P (V(P)\cap I)\right|
\le
\sum_P |V(P)\cap I|
=
\sum_P q(P).
\]
Hence $\sum_P q(P)=|I|$. Both inequalities are therefore tight. Every source path must attain its individual maximum possible $I$-occupancy; otherwise the sum would be strictly smaller than $|I|$. Moreover, equality in the union bound implies that the sets $V(P)\cap I$ are pairwise disjoint. Thus no vertex of $I$ occurs on two distinct paths, and every vertex of $I$ occurs on exactly one path. The stated structure follows.
\end{proof}

\subsection{Feasibility and filler quotas}

\begin{lemma}\label{lem:tpc-main}
If the path beginning at $p_j$ uses the cell $\alpha_{i,h}$ and target $\beta_k$, then
\[
h\ge j
\qquad\text{and}\qquad
a_i+b_h\le t_k.
\tag{9}\label{eq:tpcineq}
\]
\end{lemma}

\begin{proof}
By \cref{lem:tpc-tight}, the path must reach a target. It therefore uses the early label $2h$ on $x\to\alpha_{i,h}$, because the late label exceeds every $\alpha$-to-$y$ label. Strict temporality at $x$ gives $2j-1<2h$, equivalently $h\ge j$.

The final two labels are $C+2(a_i+b_h)$ and $C+2t_k+1$. Their strict order is equivalent to $a_i+b_h\le t_k$.
\end{proof}

Let $m_i$ be the number of cells in row $i$ selected by main paths. Then $m_i\ge0$ and $\sum_i m_i=n$.

\begin{lemma}[Prefix quotas]\label{lem:tpc-prefix}
For every $q\in[n]$,
\[
\sum_{i=1}^{q}m_i\ge q.
\tag{10}\label{eq:prefix}
\]
\end{lemma}

\begin{proof}
A filler source of type $r$ enters $x$ at time $D+2r-1$ and can use the late label $D+2i$ of a cell in row $i$ exactly when $r\le i$.

The first $q$ rows contain $qn$ cells. Main paths cover $M_q=\sum_{i=1}^q m_i$ of them. The remaining $qn-M_q$ cells must be covered by filler sources of types at most $q$, because a filler of larger type cannot reach any of these rows. There are exactly $q(n-1)$ such filler sources. Hence
\[
qn-M_q\le q(n-1),
\]
which is equivalent to \eqref{eq:prefix}.
\end{proof}

\begin{lemma}[Majorization]\label{lem:tpc-major}
Under the strict order $a_1>\cdots>a_n$, the prefix constraints imply
\[
\sum_{i=1}^{n}m_i a_i\ge\sum_{i=1}^{n}a_i.
\tag{11}\label{eq:major}
\]
Equality holds if and only if $m_i=1$ for every $i$.
\end{lemma}

\begin{proof}
Set $d_i=m_i-1$ and $D_q=\sum_{i=1}^q d_i$. By \cref{lem:tpc-prefix}, $D_q\ge0$ for $q<n$, while $D_n=0$. Summation by parts yields
\[
\sum_{i=1}^{n}(m_i-1)a_i
=
\sum_{q=1}^{n-1}D_q(a_q-a_{q+1})\ge0.
\]
Because every difference $a_q-a_{q+1}$ is strictly positive, equality holds exactly when every $D_q=0$, which is equivalent to $m_i=1$ for all $i$.
\end{proof}

\subsection{Correctness}

\begin{proof}[Proof of \cref{thm:tpc}]
Membership in NP is immediate from a certificate listing the paths and the selected edge labels.

For completeness, suppose the \DNMTS{} instance has permutations $\pi,\sigma$ satisfying $a_{\pi(j)}+b_j=t_{\sigma(j)}$. For each $j$, take the path
\[
p_j\xrightarrow{2j-1}x
\xrightarrow{2j}\alpha_{\pi(j),j}
\xrightarrow{C+2(a_{\pi(j)}+b_j)}y
\xrightarrow{C+2t_{\sigma(j)}+1}\beta_{\sigma(j)}.
\]
Exactly one cell is selected in each row. For each row $i$, match its $n-1$ filler sources bijectively to the $n-1$ unselected cells of that row, using
\[
f_{i,\ell}\xrightarrow{D+2i-1}x
\xrightarrow{D+2i}\alpha_{i,h}.
\]
These $n^2$ paths cover every vertex.

For soundness, suppose a cover of size at most $K$ exists. By \cref{lem:tpc-tight}, each main path has the form
\[
p_j\to x\to\alpha_{i_j,h_j}\to y\to\beta_{k_j},
\]
the selected cells are distinct, and $(k_1,\ldots,k_n)$ is a permutation of $[n]$. By \cref{lem:tpc-main},
\[
h_j\ge j,
\qquad
 a_{i_j}+b_{h_j}\le t_{k_j}.
\]
Summing the second inequalities gives
\[
\sum_{j=1}^{n}a_{i_j}+\sum_{j=1}^{n}b_{h_j}
\le\sum_{k=1}^{n}t_k.
\tag{12}\label{eq:tpcsum}
\]
The first term is $\sum_i m_i a_i$, which is at least $\sum_i a_i$ by \cref{lem:tpc-major}. Since $B$ is strictly increasing and $h_j\ge j$ for every $j$, we have $b_{h_j}\ge b_j$ term by term. Therefore
\[
\sum_j b_{h_j}\ge\sum_j b_j.
\]
No permutation property of the index sequence $(h_j)_{j=1}^n$ is used in this inequality; it follows solely from monotonicity of $B$ and the coordinatewise inequalities $h_j\ge j$. The permutation structure is derived only after equality is forced below.
Together with the balance identity \eqref{eq:balance}, these inequalities form
\[
\sum_i a_i+\sum_j b_j
\le
\sum_j a_{i_j}+\sum_j b_{h_j}
\le
\sum_k t_k
=
\sum_i a_i+\sum_j b_j.
\]
Hence equality holds throughout. Equality in \cref{lem:tpc-major} implies $m_i=1$ for every $i$, so the row indices $i_j$ form a permutation. Strict increase of $B$ and $h_j\ge j$ imply $h_j=j$ for every $j$. Finally, all inequalities $a_{i_j}+b_j\le t_{k_j}$ have nonnegative slack and zero total slack, so
\[
a_{i_j}+b_j=t_{k_j}
\qquad(j\in[n]).
\]
The row indices and target indices are permutations, and thus these equalities form a valid \DNMTS{} matching.
\end{proof}

\section{NP-completeness of TDPC}\label{sec:tdpc}

We prove \cref{thm:tdpc}. Sort $B$ increasingly:
\[
b_1<b_2<\cdots<b_n.
\tag{13}\label{eq:Border}
\]

\subsection{Construction}

Let
\[
M=\max\left\{\max_{i,j}(a_i+b_j),\max_k t_k\right\},
\qquad H=M+10,
\qquad L_j=(j-1)H.
\]
Create two distinguished vertices $x,y$. For each $j\in[n]$, create a source $p_j$. For each $i\in[n]$, create a middle vertex $\alpha_i$. For each $k\in[n]$, create a sink $\beta_k$. For every $j\in[n-1]$, create a guard source $u_j$ and guard sink $v_j$, and set
\[
S_j=L_j+M+6.
\]

The label sets are
\begin{align}
\lambda(p_j,x)&=\{L_j+1\}, &&j\in[n], \tag{14}\label{eq:tdpcpin}\\
\lambda(x,\alpha_i)&=\{L_h+2:h\in[n]\}, &&i\in[n], \tag{15}\label{eq:tdpcxalpha}\\
\lambda(\alpha_i,y)&=\{L_h+3+a_i+b_h:h\in[n]\}, &&i\in[n], \tag{16}\label{eq:tdpcalphay}\\
\lambda(y,\beta_k)&=\{L_h+4+t_k:h\in[n]\}, &&k\in[n], \tag{17}\label{eq:tdpcybeta}\\
\lambda(u_j,y)&=\{S_j-1\},\quad
\lambda(y,v_j)=\{S_j\}, &&j\in[n-1]. \tag{18}\label{eq:guards}
\end{align}
Set
\[
K=2n-1.
\tag{19}\label{eq:tdpcK}
\]

\begin{figure}[ht]
\centering
\begin{tikzpicture}[>=Latex,node distance=12mm and 18mm,
  v/.style={circle,draw,minimum size=7mm,inner sep=1pt},
  box/.style={draw,rounded corners,minimum width=27mm,minimum height=16mm,align=center},
  lab/.style={font=\scriptsize,align=center}]
\node[box] (p) {main sources\\$p_j$};
\node[v,right=of p] (x) {$x$};
\node[box,right=of x] (a) {middles\\$\alpha_i$};
\node[v,right=of a] (y) {$y$};
\node[box,right=of y] (b) {main sinks\\$\beta_k$};
\node[box,below=18mm of a] (u) {guard sources\\$u_j$};
\node[box,below=18mm of b] (v) {guard sinks\\$v_j$};
\draw[->] (p) -- (x);
\draw[->] (x) -- (a);
\draw[->] (a) -- (y);
\draw[->] (y) -- (b);
\draw[->] (u) -- node[above,lab] {$S_j-1$} (y);
\draw[->] (y) -- node[right,lab] {$S_j$} (v);
\end{tikzpicture}
\caption{Layered structure of the \TDPC{} construction. Guard paths occupy $y$ at separator times.}
\label{fig:tdpc}
\end{figure}
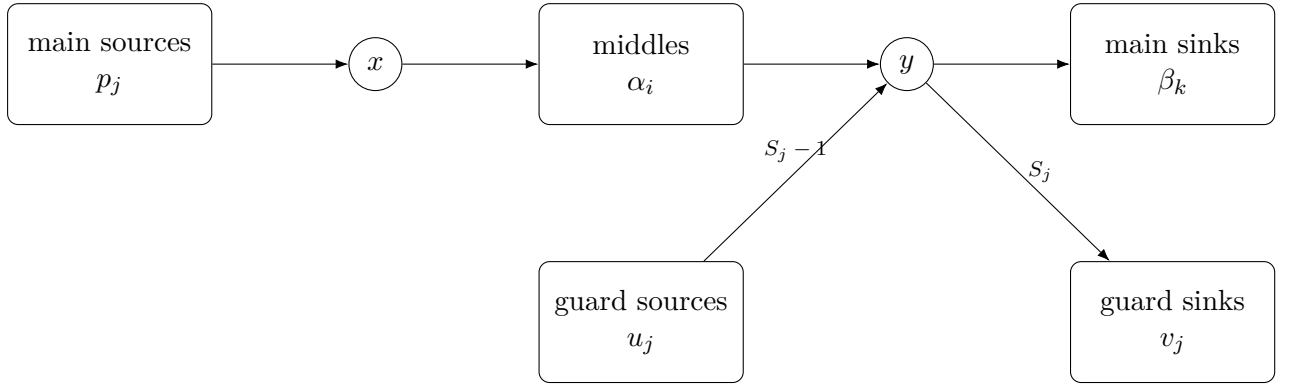

The underlying digraph is acyclic under the order
\[
\{p_j,u_j\}\prec x\prec\{\alpha_i\}\prec y\prec\{\beta_k,v_j\},
\]
with the $u_j$ placed anywhere before $y$. Every edge is incident with $x$ or $y$, so $\vc(U_G)\le2$. The undirected edges $p_1x$ and $\alpha_1y$ are disjoint, so $\vc(U_G)=2$.

\subsection{Tight structure and block forcing}

\begin{lemma}\label{lem:tdpc-tight}
Every temporally disjoint path cover of size at most $K$ has exactly $K$ paths. Every $p_j$-path contains one middle vertex and one sink, every $u_j$-path contains one sink, and every vertex outside $\{x,y\}$ occurs on exactly one path.
\end{lemma}

\begin{proof}
There are $K=2n-1$ indegree-zero sources. Let $I=V(G)\setminus\{x,y\}$. Then
\[
|I|=n+n+n+(n-1)+(n-1)=5n-2.
\]
A path beginning at $p_j$ contains at most three vertices of $I$, and a path beginning at $u_j$ contains at most two. Thus $K$ paths contain at most
\[
3n+2(n-1)=5n-2=|I|
\]
occurrences of vertices from $I$. Covering all vertices forces equality throughout, proving the claim.
\end{proof}

\begin{lemma}[First-block forcing]\label{lem:firstblock}
The path beginning at $p_j$ leaves $x$ at time $L_j+2$.
\end{lemma}

\begin{proof}
Suppose it leaves $x$ at time $L_h+2$. Strict temporality gives $h\ge j$. If $h>j$, then the path occupies $x$ at time $L_h+2$. By \cref{lem:tdpc-tight}, the distinct path beginning at $p_h$ must also leave $x$ at some time $L_r+2$ with $r\ge h$, and therefore occupies $x$ at time $L_h+2$. This contradicts temporal disjointness. Hence $h=j$.
\end{proof}

\begin{lemma}[Forced guard paths]\label{lem:guards}
For every $j\in[n-1]$, the path beginning at $u_j$ is
\[
u_j\xrightarrow{S_j-1}y\xrightarrow{S_j}v_j.
\]
\end{lemma}

\begin{proof}
By \cref{lem:tdpc-tight}, the path beginning at $u_j$ must contain a sink. It enters $y$ at time $S_j-1$. Under the TDPC occupancy convention recalled in the preliminaries, if a path enters an internal vertex at time $r$ and leaves it at time $s$, then it occupies that vertex throughout the integer interval $[r+1,s]$. Hence, regardless of which admissible outgoing edge from $y$ is eventually used, every feasible departure time $s$ satisfies $s\ge S_j$, and the $u_j$-path occupies $y$ throughout $[S_j,s]$, in particular at time $S_j$.

The unique path containing $v_j$ must use the edge $y\to v_j$ at time $S_j$, and therefore also occupies $y$ at time $S_j$. If it were distinct from the path beginning at $u_j$, temporal disjointness would be violated. Thus the two are the same path, which forces
\[
u_j\xrightarrow{S_j-1}y\xrightarrow{S_j}v_j.
\]
\end{proof}

\begin{lemma}[No backward block use at $y$]\label{lem:nobackward}
If a main path enters $y$ using a label from block $q$, then it cannot leave $y$ using a label from any block $h<q$.
\end{lemma}

\begin{proof}
Suppose a main path enters $y$ from some $\alpha_i$ using the block-$q$ label
\[
L_q+3+a_i+b_q.
\]
Any main-sink exit label belonging to an earlier block $h<q$ is at most
\[
L_h+4+t_k
\le
L_{q-1}+4+M
=
L_q-H+4+M
=
L_q-6,
\]
where $H=M+10$. Since $a_i$ and $b_q$ are positive,
\[
L_q-6 < L_q+3+a_i+b_q.
\]
Thus every exit label from an earlier block is strictly smaller than the time at which the path enters $y$, contradicting strict temporal validity. Hence an earlier-block exit is impossible.
\end{proof}

\begin{lemma}[No crossing at $y$]\label{lem:nocross}
A main path cannot enter $y$ using a label from block $q$ and leave $y$ using a label from a later block.
\end{lemma}

\begin{proof}
Every entry label in block $q$ is less than $S_q=L_q+M+6$, while every exit label in a later block is greater than $S_q$. Such a path would occupy $y$ at $S_q$, where the forced guard path from \cref{lem:guards} also occupies $y$.
\end{proof}

\subsection{Global elimination of drift}

By \cref{lem:firstblock}, the path beginning at $p_j$ leaves $x$ in block $j$. Suppose it waits at its selected middle vertex and enters $y$ in block $q_j$. The block spacing gives
\[
q_j\ge j.
\tag{20}\label{eq:qgej}
\]
By \cref{lem:nobackward,lem:nocross}, it cannot leave $y$ through either an earlier or a later block. Hence it leaves $y$ for some $\beta_{\sigma(j)}$ in the same block $q_j$. Since the middle vertices and main sinks are used exactly once after the guard paths are fixed, the selected indices define permutations $\pi$ and $\sigma$. Strict temporality at $y$ is equivalent to
\[
a_{\pi(j)}+b_{q_j}\le t_{\sigma(j)}.
\tag{21}\label{eq:tdpcineq}
\]

\begin{lemma}[Drift elimination]\label{lem:drift}
For every $j\in[n]$,
\[
q_j=j
\qquad\text{and}\qquad
a_{\pi(j)}+b_j=t_{\sigma(j)}.
\]
\end{lemma}

\begin{proof}
Summing \eqref{eq:tdpcineq} and using the fact that $\pi$ and $\sigma$ are permutations together with \eqref{eq:balance} gives
\[
\sum_{j=1}^{n}b_{q_j}\le\sum_{j=1}^{n}b_j.
\]
But \eqref{eq:qgej} and the strict increase of $B$ give $b_{q_j}\ge b_j$ term by term. Hence $q_j=j$ for every $j$. Substituting this into \eqref{eq:tdpcineq}, the total slack over all $j$ is zero. Every slack is nonnegative, so equality holds term by term.
\end{proof}

\subsection{Correctness}

\begin{proof}[Proof of \cref{thm:tdpc}]
Membership in NP follows by checking the paths, selected edge labels, coverage, and pairwise temporal disjointness.

For completeness, assume the \DNMTS{} instance has permutations $\pi,\sigma$ with $a_{\pi(j)}+b_j=t_{\sigma(j)}$. For every $j\in[n]$, use the main path
\[
p_j\xrightarrow{L_j+1}x
\xrightarrow{L_j+2}\alpha_{\pi(j)}
\xrightarrow{L_j+3+a_{\pi(j)}+b_j}y
\xrightarrow{L_j+4+t_{\sigma(j)}}\beta_{\sigma(j)}.
\]
At $x$, this path occupies only time $L_j+2$. At $y$, equality of the matched sum makes the occupancy the singleton time $L_j+4+t_{\sigma(j)}$. Different blocks are disjoint. Add the $n-1$ guard paths
\[
u_j\xrightarrow{S_j-1}y\xrightarrow{S_j}v_j.
\]
The guard time lies strictly after every main-path occupancy in block $j$ and strictly before every main-path occupancy in block $j+1$. Thus all $2n-1$ paths are pairwise temporally disjoint and cover the graph.

For soundness, assume a temporally disjoint cover of size at most $K$ exists. By \cref{lem:tdpc-tight,lem:firstblock,lem:guards,lem:nocross}, it consists of the forced guard paths and $n$ main paths. The selected middle and main sink indices form permutations $\pi$ and $\sigma$. Applying \cref{lem:drift} yields
\[
a_{\pi(j)}+b_j=t_{\sigma(j)}
\qquad(j\in[n]),
\]
which is a valid \DNMTS{} solution.
\end{proof}

\section{The boundary at vertex-cover number one}\label{sec:boundary}

The preceding results show that two vertices suffice for NP-completeness of both problems. For \TPC{}, \cref{cor:tpc-threshold} completely settles the threshold. For \TDPC{}, vertex-cover number one corresponds, up to isolated vertices, to a temporal oriented star and requires separate analysis.

Let $c$ be the center of a temporal oriented star. Write $P=\{p_1,\ldots,p_r\}$ for leaves whose edge is oriented $p_i\to c$, and $Q=\{q_1,\ldots,q_s\}$ for leaves whose edge is oriented $c\to q_j$. Let $A_i=\lambda(p_i,c)$ and $B_j=\lambda(c,q_j)$. Every two-edge temporal path has the form
\[
p_i\xrightarrow{a}c\xrightarrow{b}q_j,
\qquad a\in A_i,\ b\in B_j,\ a<b,
\]
and occupies $c$ throughout the integer interval $[a+1,b]$. We call $(i,j,a,b)$ a \emph{merge realization} and $[a+1,b]$ its \emph{occupancy interval}.

\subsection{Normalization and canonical merge families}

In a general multi-label star, temporal disjointness does not by itself forbid the same leaf from appearing in two different merge paths at different labels. The following exchange lemma shows that leaf reuse is nevertheless unnecessary in an optimum.

\begin{lemma}[De-reuse normalization]\label{lem:dereuse-star}
Let $\mathcal P$ be a feasible \TDPC{} cover of a temporal oriented star containing at least one two-edge merge path. There is a feasible cover $\mathcal P'$ with $|\mathcal P'|=|\mathcal P|$ in which every leaf belongs to at most one two-edge merge path.
\end{lemma}

\begin{proof}
Suppose a leaf $\ell$ belongs to $t\ge2$ merge paths. Keep one of them. For each of the other $t-1$ merge paths, delete that path and replace it by the singleton consisting of its other leaf endpoint. The leaf $\ell$ remains covered by the retained merge, every other deleted endpoint remains covered by its singleton, and the center remains covered because at least one merge is retained. The number of paths is unchanged. Moreover, each replacement removes center occupancy and therefore cannot create a new temporal conflict. Repeating this operation strictly decreases the number of leaf--merge incidences involving reused leaves, so the process terminates with the required cover.
\end{proof}

\begin{proposition}[Canonical form on temporal oriented stars]\label{prop:canonical-star}
If a \TDPC{} solution on a temporal oriented star contains at least one two-edge merge path, then there is an optimal solution of the same cardinality in which (i) every leaf belongs to at most one merge path, (ii) every non-merged leaf is covered by a singleton path, and (iii) no single-edge path is used. Consequently, if $M$ merges are selected, the cover has exactly $r+s-M$ paths. Thus minimizing the number of paths is equivalent to maximizing the number of leaf-unique merge realizations with pairwise disjoint occupancy intervals.
\end{proposition}

\begin{proof}
Apply \cref{lem:dereuse-star}. Once a merge path is present, the center $c$ is already covered. Replacing every remaining single-edge path by the singleton consisting of its leaf preserves the number of paths, preserves coverage, and can only remove occupancy at $c$. In the resulting cover, $M$ merges cover $2M$ leaves and every other leaf is a singleton, so the number of paths is $M+(r+s-2M)=r+s-M$.
\end{proof}

For a fixed leaf pair $(p_i,q_j)$, realizations can be pruned by interval containment.

\begin{lemma}[Dominated realizations]\label{lem:domination-star}
Fix $p_i$ and $q_j$. Suppose $(a,b)$ and $(a',b')$ are valid realizations with $a'\ge a$ and $b'\le b$. Then $(a,b)$ is dominated: every feasible canonical merge family containing $[a+1,b]$ remains feasible after replacing it by $[a'+1,b']$.
\end{lemma}

\begin{proof}
The replacement interval satisfies $[a'+1,b']\subseteq[a+1,b]$ and uses the same two leaves. Hence it creates neither a new leaf conflict nor a new center-occupancy conflict.
\end{proof}

\subsection{One-sided multi-label stars}

The single-label result extends substantially: it is enough for only one side of the oriented star to be single-label.

\begin{theorem}[One-sided multi-label stars]\label{thm:one-sided-star}
Suppose every outgoing edge $c\to q_j$ has exactly one temporal label, while the incoming label sets $A_i$ are arbitrary finite sets. Then \TDPC{} on the temporal oriented star is solvable in polynomial time by bipartite matching.
\end{theorem}

\begin{proof}
Let the distinct outgoing label values be
\[
\beta_1<\beta_2<\cdots<\beta_m,
\]
and set $\beta_0=-\infty$. Define gaps $G_k=[\beta_{k-1},\beta_k)$. Construct a bipartite graph $H$ with left side $P$ and right side $\{G_1,\ldots,G_m\}$, where
\[
p_iG_k\in E(H)
\quad\Longleftrightarrow\quad
A_i\cap G_k\neq\emptyset.
\]
We claim that the maximum number of merges in a canonical \TDPC{} solution equals $\nu(H)$.

First consider any canonical feasible merge family. Let a selected merge start at $a\in A_i$, with $a\in G_k$, and end at some outgoing label $b\ge\beta_k$. Replace its end by an outgoing edge carrying label $\beta_k$. Its occupancy interval shrinks from $[a+1,b]$ to $[a+1,\beta_k]$, so no conflict is introduced. Two temporally disjoint selected merges cannot have their start labels in the same gap $G_k$: each would contain the point $\beta_k$ in its occupancy interval. Hence distinct selected merges snap to distinct gaps. By \cref{prop:canonical-star}, they also use distinct incoming leaves. Therefore the selected merges define a matching in $H$ of the same cardinality.

Conversely, let $M$ be a matching in $H$. For every matched edge $p_iG_k$, choose any $a\in A_i\cap G_k$ and choose one outgoing leaf whose unique label is $\beta_k$. Distinct matched gaps have distinct label values, hence correspond to distinct outgoing leaves. If $k<\ell$, then a start selected from $G_\ell$ satisfies $a_\ell\ge\beta_{\ell-1}\ge\beta_k$, so the interval $[a_\ell+1,\beta_\ell]$ begins strictly after $\beta_k$, the end of the earlier interval. Thus the constructed merge intervals are pairwise disjoint. This gives a feasible canonical merge family of size $|M|$.

Therefore the maximum merge count is $\nu(H)$, and \cref{prop:canonical-star} yields an optimum cover of size $r+s-\nu(H)$.
\end{proof}

\begin{corollary}[Symmetric one-sided case]\label{cor:one-sided-symmetric}
If every incoming edge has exactly one label and the outgoing edges may have arbitrary finite label sets, then \TDPC{} on the temporal oriented star is polynomial-time solvable.
\end{corollary}

\begin{proof}
Reverse all arcs and reverse the time order, for example by replacing every label $t$ by $T-t$ for a sufficiently large constant $T$. Strict temporal paths and temporal disjointness are preserved, while incoming and outgoing roles are exchanged. Apply \cref{thm:one-sided-star}.
\end{proof}

The previously considered single-label star case is an immediate special case of \cref{thm:one-sided-star}. In that more restricted setting, the matching structure further collapses to the familiar earliest-finish greedy algorithm.

\subsection{Bounded span and an XP algorithm}

The following parameter captures how far a candidate merge reaches across labels belonging to other leaves.

\begin{definition}[Span]\label{def:span-star}
For a feasible merge realization $x=(i,j,a,b)$, define its span $\sigma(x)$ as the number of foreign leaf-label occurrences strictly between $a$ and $b$:
\[
\sigma(x)=\bigl|\{(\ell,t):\ \ell\notin\{p_i,q_j\},\ t\in\lambda(\ell),\ a<t<b\}\bigr|.
\]
Repeated numerical values on different leaves are counted as distinct occurrences. The span of the instance is
\[
\sigma_{\max}=\max_x \sigma(x),
\]
where the maximum ranges over all feasible merge realizations, including dominated ones.
\end{definition}

The next lemma shows that a large feasible merge family necessarily creates a long cross-realization somewhere in the closure of the star.

\begin{lemma}[Span bound]\label{lem:span-bound}
Let $\mathcal F$ be a feasible canonical merge family of size $k$. Then the instance contains a feasible merge realization of span at least $2k-2$. Consequently,
\[
k\le \left\lfloor\frac{\sigma_{\max}}{2}\right\rfloor+1.
\]
\end{lemma}

\begin{proof}
Write the selected realizations as
\[
x_t=(i_t,j_t,a_t,b_t),\qquad t=1,\ldots,k.
\]
By canonical feasibility, the incoming leaves $p_{i_t}$ are pairwise distinct, the outgoing leaves $q_{j_t}$ are pairwise distinct, and the occupancy intervals $[a_t+1,b_t]$ are pairwise disjoint.

Let $a_{\min}=\min_t a_t$, attained by a realization whose incoming leaf is $p_{i^*}$, and let $b_{\max}=\max_t b_t$, attained by a realization whose outgoing leaf is $q_{j^*}$. Since every selected realization satisfies $a_t<b_t$, we have $a_{\min}<b_{\max}$. Hence Cartesian closure gives the feasible cross-realization
\[
x^*=(i^*,j^*,a_{\min},b_{\max}).
\]

No two selected realizations can share the same start value: if $a_t=a_{t'}$, both occupancy intervals contain $a_t+1$. Likewise, no two selected realizations can share the same end value, since both would contain that end time. Therefore every selected start other than $a_{\min}$ lies strictly between $a_{\min}$ and $b_{\max}$, and every selected end other than $b_{\max}$ does as well.

These $2k-2$ endpoint occurrences are all foreign to $x^*$. Indeed, row uniqueness implies that every other selected start belongs to an incoming leaf different from $p_{i^*}$, and column uniqueness implies that every other selected end belongs to an outgoing leaf different from $q_{j^*}$. Incoming and outgoing leaf sets are disjoint, so cross-side numerical coincidences do not affect foreignness. Thus
\[
\sigma(x^*)\ge 2k-2.
\]
Since $\sigma(x^*)\le\sigma_{\max}$, the stated bound follows.
\end{proof}

\begin{corollary}[XP parameterized by span]\label{cor:span-xp}
\TDPC{} on temporal oriented stars is in XP parameterized by $\sigma_{\max}$. More precisely, if $\mathcal C$ is the set of feasible merge realizations, an optimum canonical merge family can be found in time
\[
|\mathcal C|^{\lfloor\sigma_{\max}/2\rfloor+1}\,\operatorname{poly}(|\mathcal C|).
\]
In particular, for every fixed span bound $\sigma_{\max}\le d$, the problem is polynomial-time solvable.
\end{corollary}

\begin{proof}
By \cref{lem:span-bound}, every feasible canonical merge family has size at most
\[
K=\left\lfloor\frac{\sigma_{\max}}{2}\right\rfloor+1.
\]
Enumerate all subsets of $\mathcal C$ of size at most $K$ and test whether the selected realizations use pairwise distinct incoming leaves, pairwise distinct outgoing leaves, and pairwise disjoint occupancy intervals. The largest feasible subset is optimal by \cref{prop:canonical-star}. The exponent depends on $\sigma_{\max}$, so this is an XP algorithm rather than an FPT algorithm.
\end{proof}

The bound is tight in the sense that $k$ pairwise disjoint merges can force a cross-realization of span exactly $2k-2$. For example, take $k$ numerically separated singleton merges with labels $a_t=2t-1$ and $b_t=2t$. The realization joining the first incoming leaf at time $1$ to the last outgoing leaf at time $2k$ has exactly the other $2k-2$ selected endpoint occurrences in its interior.

\subsection{Remaining open case}

The preceding results substantially narrow the unresolved boundary for \TDPC{} at vertex-cover number one. The problem is polynomial whenever either side of the oriented star is single-label, and every fixed maximum-span class is polynomial by \cref{cor:span-xp}. The unrestricted two-sided multi-label case remains open; here $\sigma_{\max}$ may grow with the instance size, so the XP bound does not resolve the general case.

\section{Computational validation}

The hardness reductions were supplemented by independent counterexample searches designed to target the structural lemmas rather than to replace proof. The validation strategy deliberately allowed path forms that the soundness arguments must exclude, including repeated row selections, later-column choices, unintended target assignments, block jumps, and paths terminating in the wrong sink family. The main regimes are summarized in \cref{tab:validation}.

\begin{table}[ht]
\centering
\caption{Summary of computational counterexample searches.}
\label{tab:validation}
\small
\begin{tabular}{@{}llll@{}}
\toprule
Problem & Method & Instance sizes & Outcome \\
\midrule
\TPC{} & Full temporal-path enumeration & exhaustive $n=2$, random $n=3$ & no mismatch \\
\TPC{} & Structural/Hall-condition search & exhaustive $n=3$, random $n=4,5$ & no mismatch \\
\TPC{} & Independent MILP set-cover check & enumerated small instances & no mismatch \\
\TDPC{} & Generic path/occupancy verifier & exhaustive $n=2$, random $n=2,3$ & no mismatch \\
\bottomrule
\end{tabular}
\end{table}

For \TPC{}, the reported search suite covered 1,920 balanced \DNMTS{} instances across the full-path and structural searches. The full enumerator constructs the temporal graph itself, enumerates all strict temporal paths from each source, and searches for a cover without assuming the prefix-quota or majorization lemmas. The structural verifier allows repeated selected rows, later columns, arbitrary feasible target assignments, and then checks the exact chain-graph Hall conditions for the remaining filler cells. For \TDPC{}, the verifier enumerates strict temporal paths directly, computes the exact occupancy of each vertex under the formal interval convention, and permits attempted block jumps as well as main paths ending in guard sinks and guard paths ending in main sinks.

These computations are supplementary validation only. The complexity results rest on the proofs above. The validation scripts, generated instances, and detailed test reports can be supplied as supplementary material with the submission and can also be deposited in a public repository.

\section{Conclusion}

We resolve the vertex-cover-two question left open by Cioni et al.~\cite{cioni2026}: both \TPC{} and \TDPC{} remain NP-complete on temporal DAGs whose underlying undirected graph has vertex-cover number exactly two. The two proofs require different forcing mechanisms, reflecting the fact that path overlap is permitted in \TPC{} but forbidden temporally in \TDPC{}. For \TPC{}, the result is tight: together with the polynomial-time algorithm on temporal oriented trees, it yields an exact transition from polynomial-time solvability at vertex-cover number at most one to NP-completeness at vertex-cover number two. For \TDPC{}, we sharpen the vertex-cover-one boundary beyond the single-label setting. On temporal oriented stars, leaf reuse can be eliminated without increasing the cover size; this yields a canonical merge formulation. We then show polynomial-time solvability whenever either side of the star is single-label, even with arbitrary label multiplicity on the opposite side. We also introduce a closure-inclusive span parameter and prove that any feasible family of $k$ merges forces a feasible realization of span at least $2k-2$, yielding an XP algorithm parameterized by maximum span and polynomial-time solvability for every fixed span bound. The unrestricted two-sided multi-label star case remains open.

\paragraph{Declaration of interests.} The author declares no competing interests.

\paragraph{Data and code availability.} No empirical research dataset was generated or analyzed in this study. The mathematical results are self-contained. Detailed reports from the supplementary computational counterexample searches are provided as supplementary material. The computational scripts used during exploratory validation were working research code rather than an archival software release and are available from the author upon reasonable request.

\paragraph{Declaration of generative AI and AI-assisted technologies in the manuscript preparation process.} During the preparation of this work, the author used OpenAI ChatGPT and Anthropic Claude to assist with proof auditing, adversarial counterexample searches, structural and language editing, and preparation of the manuscript. All mathematical claims, constructions, references, and final text were reviewed and verified by the author, who takes full responsibility for the content of the article.

\end{document}